\documentclass[11pt]{article}
\usepackage{times,fullpage,amsmath,amsthm,amssymb,hyperref,graphicx,color,subcaption, enumerate, lineno, bm}

\ifdefined\SODA
\else
\fi

\newtheorem{theorem}{Theorem}[section]
\newtheorem{lemma}[theorem]{Lemma}
\newtheorem{corollary}[theorem]{Corollary}
\newtheorem{problem}[theorem]{Problem}
\newtheorem{definition}[theorem]{Definition}
\newtheorem{observation}[theorem]{Observation}
\newtheorem*{utheorem}{Theorem}

\title{Simplex Range Searching Revisited:\\ How to Shave Logs in Multi-Level Data Structures}

\ifdefined\SODA
\else
\author{Timothy M. Chan%
  \thanks{Department of Computer Science,
    University of Illinois at Urbana-Champaign,
    \{tmc,dwzheng2\}@illinois.edu.  Work supported by
    NSF Grant CCF-1814026.} 
  \and Da Wei Zheng$^*$}
\fi 

\renewcommand{\pmod}[1]{{\ifmmode\text{\rm\ (mod~$#1$)}\else\discretionary{}{}{\hbox{ }}\rm(mod~$#1$)\fi}}

\def\polylog{\operatorname{polylog}}
\newcommand{\eps}{\varepsilon}
\renewcommand{\epsilon}{\varepsilon}
\newcommand{\D}{\Delta}
\newcommand{\R}{\mathbb{R}}

\newcommand{\A}{\mathcal{A}}
\newcommand{\B}{\mathcal{B}}
\newcommand{\OO}{\widetilde{O}}

\newcommand{\IGNORE}[1]{}

\newcommand{\TIMOTHY}[1]{{\color{blue} (( Timothy: #1 ))}}
\newcommand{\DAVID}[1]{{\color{red} (( David: #1 ))}}

\newcommand{\Matousek}{Matou\v{s}ek}

\newcommand{\blog}{\bm{log}\;}
\renewcommand{\L}{{\mathcal L}}

\begin{document}
\maketitle
\setcounter{page}{0}
\thispagestyle{empty}

\begin{abstract}
We revisit the classic problem of simplex range searching and related
problems in computational geometry.  We present a collection of new results which improve previous bounds by multiple logarithmic factors that were caused by the use of multi-level data structures.  Highlights include the following:
\begin{itemize}
\item
For a set of $n$ points in a constant dimension $d$,
we give data structures with $O(n^d)$ (or slightly better) space 
that can answer simplex range counting queries in optimal $O(\log n)$ time
and simplex range reporting queries in optimal $O(\log n + k)$ time, where
$k$ denotes the output size.  
For semigroup range searching, we obtain $O(\log n)$ query time with $O(n^d\mathop{\rm polylog}n)$ space.
Previous data structures with similar
space bounds by Matou\v sek from nearly three decades ago
had $O(\log^{d+1}n)$ or $O(\log^{d+1}n + k)$ query time.
\item
For a set of $n$ simplices in a constant dimension $d$,
we give data structures with $O(n)$ space 
that can answer stabbing counting queries (counting the number of simplices
containing a query point) in $O(n^{1-1/d})$ time,
and stabbing reporting queries in $O(n^{1-1/d}+k)$ time.
Previous data structures had extra $\log^d n$ factors in space and query time.
\item
For a set of $n$ (possibly intersecting) line segments in 2D,
we give a data structure with $O(n)$ space
that can answer ray shooting queries in $O(\sqrt{n})$ time.
This improves Wang's recent data structure [SoCG'20] with
$O(n\log n)$ space and $O(\sqrt{n}\log n)$ query time.
\end{itemize}
\end{abstract}

\newpage

\section{Introduction}

\paragraph{Simplex range searching.}
\emph{Simplex range searching} is among the most fundamental and central
problems in computational geometry~\cite{agarwal17,AgarwalE99,BergCKO08}.  
Its importance cannot be overstated: countless geometric algorithms  make use
of simplex range searching data structures as subroutines.
Given a set of $n$ points in
a constant dimension $d$, the goal is to build data structures so that
we can quickly find the points inside a query simplex~$q$.
Several versions exist: in \emph{range counting}, we want the number of points
inside $q$; in \emph{range reporting}, we want to report all the points inside $q$,
in time proportional to the number $k$ of output points;
in \emph{group or semigroup range query} (which generalizes range counting), we want the sum of the weights
of the points inside $q$, assuming that each input point is given a weight
from a group or semigroup.\footnote{
All groups and semigroups are assumed to be commutative in this paper.
}  Simplex ranges are fundamental because any polyhedral region can be decomposed into simplices.

After years of research, the complexity of simplex range searching is now well-understood, if we do not care about logarithmic factors.  Data structures with $O(m\polylog n)$ space and 
$O((n/m^{1/d})\polylog n)$ query time are known (with a ``$+k$'' term in the query bound for the reporting version)~\cite{Welzl89,chazelle1992,Matousek92,Matousek93}, where $m$ is a trade-off parameter between $n$ and $n^d$. These bounds are generally believed to be close to optimal.\footnote{
Notably, Chazelle~\cite{Chazelle89} proved
an $\Omega((n/\log n)/m^{1/d})$ lower bound on the query time for any $m$-space data structure in the semigroup setting; and Chazelle and Rosenberg~\cite{ChazelleR95} proved an
$\Omega(n^{1-\eps}/m^{1/d})$ lower bound on $f(n)$ for any $m$-space data structure with $O(f(n)+k)$ query time for simplex range reporting in the pointer machine model.
It is believed that the extra factors $\log n$ and $n^\eps$ are artifacts
of the proofs.  (Indeed, the $\log n$ factor disappears for $d=2$~\cite{Chazelle89}, and the $n^\eps$ factor has been slightly improved by Afshani~\cite{Afshani13}.) 
}
The trade-off is obtained by interpolating between data structures for
the two extreme cases, $m=n$ and $m=n^d$.
In fact, in the linear-space regime with $m=n$, known results have even
eliminated all of the extra logarithmic factors, i.e., there are data structures
with $O(n)$ space and $O(n^{1-1/d})$ query time~\cite{Matousek93,Chan12}.  However,
in the large-space regime with $m=n^d$, the best query time bound known 
for $O(n^d\polylog n)$ space
is $O(\log^{d+1}n)$, by Matou\v sek~\cite{Matousek93} from the 1990s.  This leads to the following question:
\begin{quote}\em
With $O(n^d\polylog n)$ space, could the
query time for simplex range searching be reduced,
ideally to $O(\log n)$?
\end{quote}

Surprisingly, no progress has been reported, despite the central importance of the simplex range searching problem.  Although the question may appear to be merely about shaving logarithmic factors, it is interesting for the following reasons:
Matou\v sek's previous solution was a \emph{multi-level} cutting tree (which we will say more about later), and there is a general feeling among researchers that the usage of multi-level data structures necessitates at least one extra logarithmic factor per level, especially when the query time is subpolynomial.  Our new results will call this rule of thumb into question.  
Secondly, once the large-space regime is improved, potentially the entire space/query-time trade-off could be improved, by combining with the known techniques in the linear-space regime.

It is not difficult to obtain $O(\log n)$ query time if
space is increased to $O(n^{d+\eps})$ for an arbitrarily small constant $\eps>0$,
but insisting on $O(n^d\polylog n)$ space is what makes the problem challenging.
Goswami, Das, and Nandy~\cite{goswami2004} showed that for $d=2$, triangle range counting (or group range searching) queries
can indeed be answered in $O(\log n)$ time with $O(n^2)$ space, but
their solution was not as good for range reporting (they obtained 
a weaker query time bound of $O(\log^2n+k)$) and did not extend to higher dimensions.
Recently, Chan and Zheng~\cite{chan2022} observed that for a certain range of
trade-offs when $m$ is between $n$ and $n^{d-\eps}$, the extra logarithmic
factors can be eliminated for some related problems, but this makes 
the question for the $m=n^d$ case all the more intriguing.

We present improved data structures for simplex range query problems in any constant dimension $d$.  With $O(n^d\polylog n)$ space,
we improve the query times to $O(\log n + k)$ for simplex range reporting, and $O(\log n)$ for simplex range counting and group or semigroup queries; these query bounds are optimal.  
In fact, for the group or reporting version of the problem, we can even reduce space to slightly below $n^d$ (by a small polylogarithmic factor).
It is straightforward to use our results to obtain improvements for the complete space/query-time trade-off as well.

    


\paragraph{Simplex range stabbing.}
Another fundamental geometric data structure problem is
\emph{simplex range stabbing}: given a set of $n$ simplices in a constant dimension~$d$, build a data structure so that we can quickly find the simplices that are stabbed by (i.e., contain) a query point.  As before, there are different versions of the problem (counting, reporting, etc.).
Range stabbing may be viewed as the ``inverse'' of range searching, where the role of input and query objects is reversed.

The complexity of simplex range stabbing is similar to simplex range searching, if we don't care about logarithmic factors.
This time, in the large-space regime, data structures with $O(n^d\polylog n)$ space and $O(\log n)$ query time follow from known techniques, but in the near-linear-space regime, known techniques give data structures with extra logarithmic factors (more precisely, $O(n\log^d n)$ space and $O(n^{1-1/d}\log^d n)$ query time~\cite{Chan12}; see also~\cite{cheng1992} for prior work on 2D triangle stabbing with similar extra logarithmic factors).  This leads to the following question:
\begin{quote}
    \em In the near-linear space regime, could the extra logarithmic factors in the space and query time for simplex range stabbing be removed?
\end{quote}
Again, there was a general feeling among researchers that these logarithmic factors might be necessary since the current solutions for simplex range stabbing were also multi-level data structures.

We show that \emph{all} of the extra logarithmic factors may be eliminated!
Specifically, with $O(n)$ space, we achieve
$O(n^{1-1/d})$ query time for the counting or group version and
$O(n^{1-1/d}+k)$ query time for the reporting version. 
In fact,  for counting or reporting, we can even reduce the query time to slightly below $n^{1-1/d}$ (by a small polylogarithmic factor) in a computational model that allows for bit packing.

\paragraph{Segment intersection searching and ray shooting.}
Lastly, we consider another related  fundamental class of geometric data structure problems, this time, about line segments in 2D\@.  Given $n$ (possibly intersecting) line segments in 2D, we want to build data structures so that we can quickly find the input line segments intersecting a query line segment (\emph{intersection searching}), or
find the first input line segment intersected by a query ray (\emph{ray shooting}).  As before, there are different versions of intersection searching (counting, group, reporting, etc.).  This class of problems has historical significance in computational geometry, having been extensively studied since the 1980s \cite{OvermarsSS90,GuibasOS88,agarwal92,cheng1992,Bar-YehudaF94,wang2020}.

The complexity of these problems are similar to triangle range searching in 2D, ignoring logarithmic factors.  Recently, in SoCG'20, Wang~\cite{wang2020} obtained improvements in the logarithmic factors
in the near-linear-space regime for the ray shooting problem: his data structure achieved $O(n\log n)$ space and $O(\sqrt{n}\log n)$ query time.
There was still an extra logarithmic factor in both space and time.

We obtain a new data structure that 
eliminates \emph{both} logarithmic factors.  With $O(n)$ space, we achieve $O(\sqrt{n})$ query time, not just for ray shooting but also for segment intersection counting or searching in the group setting, or reporting (with a ``$+k$'' term for reporting).
In contrast, Wang's method did not extend to intersection counting.  Our results even improve over previous specialized results for nonintersecting segments.
In fact, for counting or reporting, we can even reduce the query time to slightly below $\sqrt{n}$ (by a small polylogarithmic factor).


Our new results are summarized in Table~\ref{tab:results}.
(The $\OO$ notation hides polylogarithmic factors throughout this paper.)

\begin{table}[t]
  \begin{center}
    \begin{tabular}{|c|c|c|c|}\hline
      Problem                       & Space                  & Query time      & Ref.\\ \hline
      simplex reporting              & $n^d$  & $\log^{d+1} n + k$ & \cite{Matousek93}\\
      & $n^d$ &   $\log n + k$ & new\\ 
      simplex counting (or group)     & $n^d$          & $\log^{d+1} n$   & \cite{Matousek93} \\
      &  $n^d$ &   $\log n$ & new \\
      simplex semigroup  & $n^d$ & $\log^{d+1} n$  & \cite{Matousek93} \\
      & $\OO(n^d)$ &   $\log n$ & new \\ \hline
      simplex stabbing reporting & $n\log^d n$ & $n^{1-1/d}\log^d n + k$ & \cite{Chan12}\\
      & $n$     &   $n^{1-1/d}+k$  & new\\
      simplex stabbing counting (or group)
      & $n\log^d n$ & $n^{1-1/d}\log^d n$ & \cite{Chan12}\\
      & $n$     &   $n^{1-1/d}$  & new\\\hline
      segment intersection reporting in 2D  & $n\log^2 n$ & $\sqrt{n}\log^2n + k$ & \cite{cheng1992}\\
      & $n$     &   $\sqrt{n}+k$ & new         \\
      segment intersection counting (or group) in 2D & $n\log^2n$ & $\sqrt{n}\log n$ & \cite{Bar-YehudaF94}\\
      & $n$     &   $\sqrt{n}$ & new         \\
      segment ray shooting in 2D & $n\alpha(n)\log^2n$ & $\sqrt{n\alpha(n)}\log n$ & \cite{Bar-YehudaF94}\\
      & $n\log^2n$ & $\sqrt{n}\log n$ & \cite{cheng1992}\\
      & $n\log n$ & $\sqrt{n}\log n$ & \cite{wang2020}\\
      & $n$     &   $\sqrt{n}$   & new      \\\hline
   \end{tabular}
    \caption{Summary of new results and selected previous results. (In some of the new results, we can even get slightly below $n^d$ space, or slightly below $n^{1-1/d}$ or $\sqrt{n}$ query time.) }
  \label{tab:results}
  \end{center}
\end{table}

\paragraph{Techniques.}
In the 1980s and 1990s, a number of techniques were developed by computational geometers for solving problems related to simplex range searching---most notably, \emph{cutting trees} \cite{Clarkson88,ChazelleF90,Chazelle04,AgarwalE99,agarwal17} in the large-space regime, and \emph{partition trees} \cite{Welzl89,Matousek92,Matousek93,AgarwalE99,agarwal17} in the linear-space regime.  Cutting trees work naturally for \emph{halfspace} range searching, but to extend the solution to \emph{simplex} range searching, one needs to apply a standard
\emph{multi-level} technique~\cite{Matousek93,AgarwalE99,agarwal17} which causes extra logarithmic factors.
For example, in a 2-level data structure, we have a ``primary'' (outer) tree structure, where each node stores a ``secondary'' (inner) tree structure solving some intermediate subproblem for a subset of the input called a ``canonical subset''.
The classic example of a multi-level data structure is the \emph{range tree}~\cite{AgarwalE99,BergCKO08} for $d$-dimensional orthogonal range searching, which has
about $d$ levels and has $\log^{d-O(1)}n$ factors in both space and query time (and there are known lower bounds suggesting that these factors are necessary for orthogonal range searching under various computational models~\cite{Chazelle90a,Chazelle90b}).

Similarly, for simplex range stabbing (and also line-segment intersection searching and ray shooting) in the near-linear-space regime, one needs a multi-level version of the partition tree, which explains the extra factors in all the previous results.

Our new data structures will still be based on the same standard techniques of cutting trees, partition trees, and multi-leveling, but the novelty lies in how to combine them.
The following (loosely stated) principle will be the key: 
\begin{quote} \em
In a multi-level data structure, if the secondary structures have strictly lower complexity (in space or time) than the primary structures, then the overall complexity do not increase by logarithmic factors.  
\end{quote}
For example, if $S_0(n)$ denotes the cost (say, space) of the secondary structure, and if
    the cost of the primary structure satisfies a recurrence of the form
    $S(n)=aS(n/b) + S_0(n)$ where $S_0(n)\ll n^{\log_b a-\eps}$, then
    $S(n) = O(n^{\log_b a})$ without extra logarithmic factors, according to the master theorem.
This simple observation is hardly original, but its power seems to have been overlooked, at least in the context of simplex range searching.  It suggests that it is advantageous to rearrange the levels of a multi-level data structure so that the innermost-level structures solve intermediate subproblems that
have strictly lower complexity.

For simplex range searching, we first decompose the query simplex into
subcells where all but two sides/facets are vertical;
in fact, for counting or group range  queries, we may assume that only one side is nonvertical,
by the use of subtraction (this trick was used before for triangle range searching in 2D, e.g., \cite{Agarwal90ii}).  Following the above principle, we can let the innermost levels of the data structure handle the vertical sides, since they project to range searching in $d-1$ dimensions, which has strictly lower complexity.  This way, we can easily obtain $O(\log^2n)$ query time (already a substantial improvement over $O(\log^{d+1}n)$).  It turns out that the final extra $\log n$ factor can be avoided as well by another standard trick:
\begin{quote}\em
Use a tree with a nonconstant branching factor $n^\eps$ at each node with subtree size $n$. 
\end{quote}
This idea is also not original, and is well known; for example,
some versions of Matou\v sek's original partition trees~\cite{Matousek92} already used this choice of branching factor.
    This choice of branching factor not only leads to a tree with smaller $O(\log\log n)$ height, but allows the
query cost to be bounded by a geometric series: a recurrence of the form $Q(n)=Q(n^{1-\eps})+O(\log n)$ solves to $Q(n)=O(\log n)$.
The complete solution (see Section~\ref{sec:group}) is conceptually quite 
simple---in hindsight, it is surprising that it was missed before.

For reporting or semigroup range queries, subtraction is forbidden, and so we need to deal with two nonvertical sides.  We could add an extra level and obtain
$O(\log^2n)$ query cost, but we can do better:  
\begin{itemize}
    \item 
For reporting queries (see Section~\ref{sec:report}), we let the innermost level handle one nonvertical side (defined by a halfspace), since halfspace range reporting is known to have strictly lower complexity~\cite{Matousek92CGTA}; we then let the next levels handle all the
vertical sides, as the complexity is still strictly lower than general simplex range searching; we finally let the primary level handle the remaining nonvertical side.  
\item
For semigroup range queries (see Section~\ref{sec:semigroup}), we face more challenges---this is the most technically intricate part of the paper. We first observe that if we only account for the cost of semigroup arithmetic operations rather than actual running time, it is possible to achieve \emph{constant} query cost for the subproblems at the innermost level, and this eventually leads to logarithmic cost for the original problem with some careful choice of branching factors.  
However, the actual running time includes the cost of point location operations, which is at least logarithmic per subproblem.  We show that for the multiple point location subproblems encountered here, we can solve each in constant time.  This is inspired by recent work of Chan and Zheng~\cite{chan2022} on 2D fractional cascading in arrangements of lines, although in our application, fractional cascading turns out to be unnecessary---instead, it suffices to set up appropriate pointers between faces in  different arrangements.
\end{itemize}

Our data structures for simplex stabbing (see Section~\ref{sec:stab}) are based on similar ideas.  However, there is an extra $\log\log n$ factor in the space bound due to the fact that the outermost level has tree height $O(\log\log n)$.
We use additional bit-packing tricks to remove this extra factor, which may be of independent interest, as we have not seen such tricks used before for nonorthogonal range searching, though they are common in the literature on orthogonal range searching.  
(Note that there is no cheating here; we just assume $\log n$ bits may be fit in a single word.  We do not use bit packing on the group elements.
See the beginning of Section~\ref{sec:stab:bit} for more details on the computation model.  Even without these tricks, our weaker $O(n\log\log n)$ space bound is still a significant improvement over previous bounds.)

Our data structures for segment intersection and ray shooting 
(see Section~\ref{sec:seg}) 
are also based on similar ideas.  Many of the previous methods used multi-level data structures (such as \cite{cheng1992,Bar-YehudaF94,wang2020}) where the outermost level is a \emph{segment tree}~\cite{BergCKO08} (essentially a 1D structure) and the inner levels are partition trees.  The subproblems that arise in the inner levels are then special cases of the problems when the input objects or query objects are rays or lines.
Following the abovementioned principle, we creatively rearrange the levels, like in our simplex range searching or range stabbing data structures.
Perhaps the reason that this idea was missed before is that with the levels rearranged, the subproblems are not as natural to state geometrically.  Still, the idea is quite simple, in hindsight.


We hope that our ideas will find many more applications in improving other multi-level data structures in computational geometry.

\section{Simplex Range Searching}
Given a set $P$ of $n$ points in $\R^d$,
the goal of simplex range searching is to build (static) data structures so that we can quickly count
the points of $P$ inside  
a query simplex, or report them, or compute the sum of their weights from a group or semigroup.


\subsection{Preliminaries}\label{sec:prev}

We begin by reviewing known techniques used in previous data structures for simplex range searching.
\IGNORE{
\paragraph{Cuttings.}
We first consider the case 
where the query range is a single upper halfspace bounded by a hyperplane $h$.
By \emph{geometric duality}~\cite{BergCKO08}, the points of $P$ becomes a set of hyperplanes $H$. The query hyperplane $h$ becomes a dual point $h^*$
and the points of $p$ lying above $h$ is the set of hyperplanes in $H$ lying above $h^*$.
We will see how this dual problem can be solved with cuttings.
}
Let $H$ be a collection of $n$ hyperplanes in $\R^d$. 
For a parameter $1\le r\le n$, a $(1/r)$-\emph{cutting} of $H$ \cite{ChazelleF90,Chazelle04} is a collection $\Gamma$ of (possibly unbounded) simplices called cells with the following properties:
\begin{enumerate}
\item The cells of $\Gamma$ are interior disjoint and cover all of $\R^d$.
\item At most $n/r$ hyperplanes of $H$ intersect any simplex $\Delta \in \Gamma$.
\end{enumerate}

The \emph{conflict list} $H_\Delta$ is the set of hyperplanes of $H$ that intersect any $\Delta \in \Gamma$.
The size of a cutting is the number of cells. The following theorem by Chazelle \cite{Chazelle93} is the best known result for computing $(1/r)$-cuttings.

\begin{lemma}[Cutting Lemma] \label{lem:cutting}
  Let $H$ be a set of $n$ hyperplanes over $\R^d$. For any $r\le n$, it is possible to 
  compute a $(1/r)$-cutting of size $O(r^d)$ in time $O(nr^{d-1})$.
  Furthermore, in the same time bound, we can compute all the conflict lists, as well as a point location structure so that
  we can determine the cell containing any given point in $O(\log r)$ time.
\end{lemma}

The cutting lemma is very useful for doing divide-and-conquer in computational geometry.
We can use the lemma recursively to build a multi-level data structure for semigroup simplex range searching as follows:

Consider the case when the query range is 
an intersection of $j$ halfspaces.  Call this a ``level-$j$'' query.
Simplex range searching corresponds to  level-$(d+1)$ queries.
Level-0 queries can trivially be solved in $O(1)$ time (by just storing the sum of the weights of all input points).

By \emph{geometric duality}~\cite{BergCKO08}, the input point set $P$ becomes a set of hyperplanes $H$. A hyperplane $h$ bounding a query halfspace (say it is an lower halfspace) becomes a dual point $h^*$, 
and the points of $p\in P$ lying above $h$ correspond to dual hyperplanes in $H$ lying above $h^*$.

We apply the cutting lemma to the dual hyperplanes $H$.
For each cell $\Delta$ of the cutting, we recurse on the conflict list $H_\Delta$.
In addition, we let $C^+_\Delta$ (resp.\ $C^-_\Delta$) be the subset of the hyperplanes in $H$ that are completely above (resp.\ below) $\Delta$;
we store this subset (called a \emph{canonical subset}) in a data structure for level-$(j-1)$ queries.
This gives the following recurrence for the space (and also preprocessing time) of the data structure:
\[ S_j(n) \:=\: O(r^d)S_j(n/r) + O(r^d) S_{j-1}(n) + O(nr^d). \]
To answer a level-$j$ query, let $h$ be the hyperplane bounding one of its $j$ halfspaces; w.l.o.g., assume that this halfspace is an upper halfspace.
We find the cell $\Delta$ containing the dual point $h^*$,
recursively answer the level-$j$ query for $H_\Delta$,
and answer a level-$(j-1)$ query for $C^+_\Delta$
(since we already know that $h^*$ is above all hyperplanes in $C^+_\Delta$, one of the $j$ halfspaces can be dropped from this query); we then return the sum.
This gives the following recurrence for the query time:
\[ Q_j(n) \:=\: Q_j(n/r) + Q_{j-1}(n)  + O(\log r). \]
Choosing $r$ to be a sufficiently large constant immediately gives 
a data structure with space $S_{d+1}(n)=O(n^{d+\eps})$
and query time $Q_{d+1}(n)=O(\log^{d+1} n)$.

Matou\v sek \cite{Matousek93} used a hierarchical version 
of the cutting lemma by Chazelle~\cite{Chazelle93}, which we will discuss in Section~\ref{sec:hier},
to reduce the space bound to $\OO(n^d)$, while keeping query time $O(\log^{d+1}n)$.

Alternatively, we can choose $r=N^\eps$ where $N$ is the global input size,
to ensure that the recursion depth is $O(1)$; this gives query time $O(\log N)$ and space $O(N^{d+O(\eps)})$ 
(which can be rewritten as $O(N^{d+\eps})$ by
readjusting $\eps$ by a constant factor).
As we will use this result later, we state it as a lemma
(the extra $N^\eps$ factor turns out to be tolerable since we will apply this lemma only in $d-1$ dimensions):

\begin{lemma} \label{lem:extra_large_space_ds}
  Simplex range searching on $n$ points in $\R^{d}$ with weights from a semigroup can be performed with 
  $O(\log n)$ query time using a data structure with $O(n^{d+\eps})$ space and preprocessing time
  for any fixed $\eps> 0$.
\end{lemma}


\subsection{Group simplex range searching}\label{sec:group}

We now present our new data structure for simplex range searching in the group setting (which in particular is sufficient for counting).
We first introduce the following subproblems:

\begin{definition} \label{def:simplex_level2}
For $i\in\{0,1,2\}$ and $j\in\{0,\ldots,d\}$, an \emph{$(i,j)$-sided range} refers
the intersection
of $i$ arbitrary halfspaces and $j$ vertical halfspaces.
Here, ``vertical'' means ``parallel to the $d$-th axis'' (thus projecting a vertical halfspace along the $d$-th axis yields a halfspace in $d-1$ dimensions).
When the query is an $(i,j)$-sided range, we refer to it as an \emph{$(i,j)$-sided query}.
\end{definition}

\begin{observation}\label{obs1}
A simplex range query reduces to a constant number of  $(2,d)$-sided queries.
\end{observation}
\begin{proof}
  Just take the vertical decomposition~\cite{AgarwalS00} of the query simplex.  Since a simplex has $O(1)$ complexity, the decomposition gives $O(1)$ cells, where each cell has two nonvertical facets.  We may assume that the projection of the cell along the $d$-th axis is a $(d-1)$-dimensional simplex (if not, we can triangulate the projection).  We can answer a range query for each cell and return the sum of the answers.
\end{proof}

\begin{observation}\label{obs2}
In the group setting, a $(2,d)$-sided query reduces to two $(1,d)$-sided queries.
\end{observation}
\begin{proof}
  This follows by subtraction, since a cell with 2 nonvertical facets can be 
  expressed as the difference of two cells each with 1 nonvertical facet.
\end{proof}

\begin{theorem} \label{thm:simplex_group_polylog}
  Simplex range searching on $n$ points in $\R^{d}$ with weights from a group can be performed with
  $O(\log n)$ query time using a data structure with $\OO(n^{d})$ space and preprocessing time.
\end{theorem}
\begin{proof}
  By the above two observations, a simplex range query reduces to a constant number of $(1,d)$-sided queries in the group setting.  It suffices to describe a solution to $(1,d)$-sided queries.
  
  Like before, we apply the cutting lemma to the dual hyperplanes $H$.
For each cell $\Delta$ of the cutting, we recurse on the conflict list $H_\Delta$.
In addition, we let $C^+_\Delta$ (resp.\ $C^-_\Delta$) be the subset of the hyperplanes in $H$ that are completely above (resp.\ below) $\Delta$;
we store this subset (a \emph{canonical subset}) in a data structure for $(0,d)$-sided queries.  Note that $(0,d)$-sided queries are equivalent to
range queries on the $(d-1)$-dimensional vertical projection of the input points, and so there is a data structure with  $S_{0,d}(n)=O(n^{d-1+\eps})$ space (and preprocessing time), and
$Q_{0,d}(n)=O(\log n)$ query time by Lemma~\ref{lem:extra_large_space_ds}.
This gives the following recurrence for the space (and also preprocessing time) of the data structure:
  \[ S_{1,d}(n) \:=\: O(r^d)S_{1,d}(n/r) + O(r^d) S_{0,d}(n) + O(nr^d). \]
To answer a $(1,d)$-sided query, let $h$ be the hyperplane bounding its nonvertical halfspace; w.l.o.g., assume that this halfspace is an lower halfspace.
We find the cell $\Delta$ containing the dual point $h^*$,
recursively answer the query for $H_\Delta$,
and answer a $(0,d)$-sided query for $C^+_\Delta$
(since we already know that $h^*$ is above all hyperplanes in $C^+_\Delta$, all the remaining sides are vertical); we then return the sum.
This gives the following recurrence for the query time:
\[ Q_{1,d}(n) \:=\: Q_{1,d}(n/r) + Q_{0,d}(n) + O(\log r). \]

  We choose $r=n^\eps$
  for a sufficiently small constant $\eps$, 
  and plug in $S_{0,d}(n)=O(n^{d-1+\eps})$. Then
    \[ S_{1,d}(n) \:=\: O(n^{\eps d})S_{1,d}(n^{1-\eps}) + O(n^{d-1+O(\eps)}). \]
  The recursion has $O(\log \log n)$ depth.
  Due to a constant-factor blowup, we 
  get $S_{1,d}(n) = O(n^d 2^{O(\log \log n)}) = \OO(n^d)$.
  On the other hand,
    \begin{align*}
    Q_{1,d}(n) &\ \le\ Q_{1,d}(n^{1-\eps}) + O(\log n) \\
         &\ \le\ O(\log n + (1-\eps)\log n + (1-\eps)^2 \log n + \cdots) 
         \ \le\ O(\log n)
  \end{align*}
  by a geometric series.
\end{proof}

\subsection{Reducing space from $\OO(n^d)$ to $O(n^d)$ (or better)}\label{sec:hier}

In the proof of Theorem~\ref{thm:simplex_group_polylog}, we get an extra polylogarithmic factor in space because of the constant-factor blowup, but this can be avoided by using a known ``hierarchical'' version of cuttings due to Chazelle~\cite{Chazelle93}:

\begin{lemma}[Hierarchical Cutting Lemma]\label{thm:hierarchical_cuttings}
  Let $H$ be a set of $n$ hyperplanes over $\R^d$. For any sequence
  $r_1 < r_2<\cdots < r_\ell\le n$,
  we can compute a tree of cells, such that the cell of each node is the disjoint union of the cells of its children, and
  the cells at each depth $i$ form a $(1/r_i)$-cutting of $H$ of size $O(r_i^d)$.  All the cells and all the conflict lists can be computed
  in  time $O(nr_\ell^{d-1})$, and we can determine the child cell  containing any given point at a node of depth $i$ in time $O(\log (r_{i+1}/r_i))$.
\end{lemma}

Originally, Chazelle proved the above lemma for the sequence $1,\rho,\rho^2,\ldots$ for some constant $\rho>1$, but the above generalization follows immediately by rounding each $r_i$ to a power of $\rho$, and keeping only the nodes with depths in a subsequence of $1,\rho,\rho^2,\ldots$ (thereby compressing the tree).

\begin{theorem} \label{thm:simplex_group}
  Simplex range searching on $n$ points in $\R^{d}$ with weights from a group can be performed with
  $O(\log n)$ query time using a data structure with $O(n^{d})$ space and preprocessing time.
\end{theorem}
\begin{proof}
To reduce space in our data structure, we just replace all the cuttings in the proof of Theorem~\ref{thm:simplex_group_polylog} (the outer level) with the tree of cuttings from the hierarchical cutting lemma,
using the sequence $r_i=n/n_i$, $n_0=n$, $n_{i+1}=n_i^{1-\eps}$, and $\ell=O(\log\log n)$.
The space bound becomes
$S_{1,d}(n) = O\left(\sum_{i=0}^{\ell-1} r_{i+1}^d n_i^{d-1+O(\eps)} \right)
= O\left(\sum_{i=0}^{\ell-1} n^d/ n_i^{1-O(\eps)}\right)$,
which sums to $O(n^d)$.
\end{proof}

\paragraph{Remark.}
In fact, space can be further reduced to slightly below $n^d$ as follows:
We set $\ell$ so that  $r_\ell=n/A$ for a parameter $A$, and switch to a different data structure for leaf subproblems of size~$A$ with
$O(m\log^{O(1)}A)$ space and 
$O(A/m^{1/d})$ query time~\cite{Matousek93,Chan12}; by choosing $m=A^{d-\delta}$ for a sufficiently
small constant $\delta>0$, we get $O(A^{d-\delta}\log^{O(1)}A)$ space and $O(A^{\delta/d})$ query time.
The space summation now gives
$O(n^d/A^{1-O(\eps)})$, and the total space of the leaf structures is
$O((n/A)^d\cdot A^{d-\delta}\log^{O(1)}A)=O((n^d/A^{\delta})\log^{O(1)}A)$, which dominates the sum.
The query time is $O(\log n +A^{\delta/d})$.
Setting $A=\log^{d/\delta}n$ yields logarithmic query time and $O((n^d/\log^d n)(\log\log n)^{O(1)})$ space.

Note that this type of $O(n^d/\log^{\Omega(1)}n)$ space bound is not a complete surprise and has appeared before for certain problems (e.g., see~\cite{ChazelleF94}).

The small extra $\log\log n$ factors are likely improvable by bootstrapping.  (In fact, for counting, we might even be able to get slightly below $n^d/\log^dn$ by bit packing tricks; see the remark at the end of  Section~\ref{sec:stab:bit}.)


\subsection{Simplex range reporting}\label{sec:report}

For simplex range reporting, the subtraction trick in Observation~\ref{obs2} is no longer applicable.  We could
add another level to our multi-level data structure to reduce
$(2,d)$-sided queries to $(1,d)$-sided queries, but this would result in an extra logarithmic factor in the query time.  We propose a different strategy to solve the $(2,d)$-sided problem.

We start with the $(1,0)$-sided query problem, which can be solved by
known results for halfspace range reporting~\cite{clarkson87,Matousek92CGTA}, as stated in the following lemma.
(This data structure was obtained by a recursive application of a ``shallow'' variant of cuttings. The extra $n^\eps$ factor in the space bound was later improved
by \Matousek{} and Schwarzkopf \cite{MatousekS93}, but this weaker result is enough as what's important is
that the exponent is strictly less than $d$.)


\begin{lemma}\label{lem:halfspace_range_reporting}
Halfspace reporting on $n$ points in $\R^d$ 
can be performed with $O(\log n+k)$ query time using a data structure
with $O(n^{\lfloor d/2\rfloor+\eps})$ space and preprocessing time for any fixed $\eps>0$.
\end{lemma}

Next we solve the $(1,d)$-sided query problem:

\begin{lemma} \label{lem:simplex_reporting_level2}
There is a data structure for $(1,d)$-sided queries in $\R^d$ with $O(\log n+k)$ query time and $O(n^{d-1+\eps})$ space and preprocessing time for any fixed $\eps > 0$.
\end{lemma}
\begin{proof}
Applying the same method as in Section~\ref{sec:prev} using cuttings in $d-1$ dimensions to deal with the vertical sides, 
we obtain a data structure for $(1,j)$-sided queries with the following recurrences for space and query time (ignoring the ``$+k$'' reporting cost):
\[ S_{1,j}(n) \:=\: O(r^{d-1})S_{1,j}(n/r) + O(r^{d-1}) S_{1,j-1}(n) + O(nr^{d-1}) \]
\[ Q_{1,j}(n) \:=\: Q_{1,j}(n/r) + Q_{1,j-1}(n)  + O(\log r), \]
with  $S_{1,0}(n) = O(n^{\lfloor d/2\rfloor+\eps}) \le O(n^{d-1+\eps})$ and
$Q_{1,0}(n)=O(\log n)$ by Lemma~\ref{lem:halfspace_range_reporting} for the base case.

We choose $r=N^\eps$ where $N$ is the global input size (``global'' is with respect to this lemma), 
to ensure that the recursion depth is $O(1)$. This gives $S_{1,d}(N)=O(N^{d-1+O(\eps)})$ and $Q_{1,d}(N)=O(\log N)$ (as usual, we can readjust $\eps$ by a constant factor).
\IGNORE{
Construct the data structure of Theorem~\ref{thm:extra_large_space_ds}
using $O(n^{d-1+\eps})$ space
for the projection of all points to the hyperplane $x_0 = 0$,
and for every canonical subset, build the data structure of Theorem~\ref{thm:halfspace_range_reporting} using $O(n^{\lfloor{d}/{2}\rfloor+\eps})$ space.
Since the data structure of Theorem~\ref{thm:extra_large_space_ds} has $O(1)$ levels,
the query time is $O(\log n + k)$ to find the canonical subsets for (a) and report all points satisfying (b). 
The space usage is $O(n^{d-1+\eps})$.
}
\end{proof}

Finally, we solve the $(2,d)$-sided query problem:

\begin{theorem} \label{thm:simplex_reporting}
Simplex reporting on $n$ points in $\R^d$
can be performed with $O(\log n + k)$ query time (where $k$ is the output size) using a data structure 
with $O(n^d)$ space and preprocessing time.
\end{theorem}
\begin{proof}
By Observation~\ref{obs1}, it suffices to solve the $(2,d)$-sided query problem.
As in the proof of Theorem~\ref{thm:simplex_group_polylog},
we obtain a data structure  with the following recurrences for space and query time  (ignoring the ``$+k$'' reporting cost):
  \[ S_{2,d}(n) \:=\: O(r^d)S_{2,d}(n/r) + O(r^d) S_{1,d}(n) + O(nr^d) \]
\[ Q_{2,d}(n) \:=\: Q_{2,d}(n/r) + Q_{1,d}(n) + O(\log r), \]
with $S_{1,d}(n)=O(n^{d-1+\eps})$ and $Q_{1,d}(n)=O(\log n)$ by Lemma~\ref{lem:simplex_reporting_level2}.

  We choose $r=n^\eps$
  for a sufficiently small constant $\eps$.
  As in the proof of Theorem~\ref{thm:simplex_group_polylog}, the recurrences solve to 
  $S_{2,d}(n) = \OO(n^d)$
  and $Q_{2,d}(n)=O(\log n)$.
  
  As in Section~\ref{sec:hier}, the space bound in the above data structure can be reduced from $\OO(n^d)$ to
$O(n^d)$, by using hierarchical cuttings for the outermost level.
\end{proof}

To summarize, in the above multi-level data structure, the innermost level tackles one of the nonvertical sides (solving a $(1,0)$-sided problem); the next levels then add all the vertical sides (solving a  $(1,d)$-sided problem); the outermost level finally adds the second nonvertical side (solving a  $(2,d)$-sided problem).  We emphasize that this unusual order is critical to achieving our result.

As in the remark in Section~\ref{sec:hier},
space can be further reduced to $O(n^d/\log^{\Omega(1)} n)$.

\IGNORE{
*******

We can only use this result on halfspace range reporting if 
it is at the inner-most level of the multilevel 
data structure. Thus we can consider the problem of reporting points 
bounded by $d$ vertical hyperplane and below one hyperplane:
\begin{problem} \label{prob:simplex_reporting_level2}
For $n$ points in $\R^d$ we can build a data structure to handle queries of the following form:
Given an arbitrary halfspace $H$ in $\R^d$ and $d$ vertical halfspaces $h_1, ..., h_d$ whose projection onto $x_0 = 0$ is a simplex,
report the points that lie within the intersection of these $d+1$ halfspaces.
\end{problem}

For the general simplex reporting problem, we can build a hierarchical cutting data structure
for the upper halfplane, and build the data structure of Lemma~\ref{lem:simplex_reporting_level2} 
for every canonical subset in the same manner as Thoerem~\ref{thm:simplex_group}. The analysis
for runtime and query time follows by the same argument (with an extra ``+k'' term).

}

\subsection{Semigroup simplex range searching}\label{sec:semigroup}

For simplex range searching in the semigroup model,
the subtraction trick in Observation~\ref{obs1} is again not applicable, but we can still obtain logarithmic query time by a more intricate solution.

\subsubsection{Bounding the number of semigroup operations}

\newcommand{\AAA}{{\cal A}}
To warm up, we first relax the computational model and measure the query complexity by  the number of semigroup sum operations instead of actual running time.   All other operations not involving the input weights
(like point location) are ``free''.

We first observe that it is possible to improve the query complexity of Lemma~\ref{lem:extra_large_space_ds} to \emph{constant} in this setting.

\begin{lemma} \label{lem:semigroup1}
  There is a data structure for simplex range searching in $\R^d$ in the semigroup setting with $O(n^{d+\eps})$ space, such that a query can be answered using $O(1)$ semigroup operations.
\end{lemma}
\begin{proof}
We use the same method as in Section~\ref{sec:prev}.
Because point location is free, the recurrences now become
\[ S_j(n) \:=\: O(r^d)S_j(n/r) + O(r^d) S_{j-1}(n) + O(nr^d) \]
\[ Q_j(n) \:=\: Q_j(n/r) + Q_{j-1}(n)  + O(1). \]
As we have chosen $r=N^\eps$ so that the recursion depth is $O(1)$, we now have $Q_{d+1}(N)=O(1)$ (and $S_{d+1}(N)=O(N^{d+O(\eps)})$ as before). 
\end{proof}

We can solve the $(0,d)$-sided query problem by applying the above lemma in $d-1$ dimensions.  Next, we can solve the $(1,d)$-sided query problem as in Theorem~\ref{thm:simplex_group_polylog}, but now with $O(\log\log n)$ query complexity:


\begin{lemma}\label{lem:semigroup2}
There is a data structure for $(1,d)$-sided queries in $\R^d$ in the semigroup setting with $\OO(n^d)$ space, such that a query can be answered in $O(\log\log n)$ semigroup operations.
\end{lemma}
\begin{proof}
We use the same method as in the proof of Theorem~\ref{thm:simplex_group_polylog}.
Because point location is free, the recurrences now become
  \[ S_{1,d}(n) \:=\: O(r^d)S_{1,d}(n/r) + O(r^d) S_{0,d}(n) + O(nr^d) \]
\[ Q_{1,d}(n) \:=\: Q_{1,d}(n/r) + Q_{0,d}(n) + O(1), \]
with   $S_{0,d}(n)=O(n^{d-1+\eps})$ and $Q_{0,d}(n)=O(1)$ by Lemma~\ref{lem:semigroup1}. 
As we have chosen $r=n^\eps$, the query recurrence 
\[ Q_{1,d}(n) = Q_{1,d}(n^{1-\eps}) + O(1)
\]
now gives $Q_{1,d}(n)=O(\log\log n)$ (and
$S_{1,d}(n)=\OO(n^d)$ as before).
\end{proof}

Finally, we can solve the $(2,d)$-sided query problem with $O(\log n)$ query complexity by using a polylogarithmic branching factor:

\begin{lemma}\label{lem:semigroup3}
There is a data structure for simplex range queries in $\R^d$ in the semigroup setting with $\OO(n^d)$ space, such that a query can be answered in $O(\log n)$ semigroup operations.
\end{lemma}
\begin{proof}
By Observation~\ref{obs1}, it suffices to solve the $(2,d)$-sided query problem.
We reduce $(2,d)$-sided queries to $(1,d)$-sided queries as in the proof of Theorem ~\ref{thm:simplex_group_polylog},
which gives the recurrences
  \[ S_{2,d}(n) \:=\: O(r^d)S_{2,d}(n/r) + O(r^d)S_{1,d}(n) + O(nr^d) \]
\[ Q_{2,d}(n) \:=\: Q_{2,d}(n/r) + Q_{1,d}(n) + O(\log r), \]
where $S_{1,d}(n)=\OO(n^d)$ and $Q_{1,d}(n)=O(\log\log n)$ by Lemma~\ref{lem:semigroup2}.

We choose $r=\log^\eps n$.  
The recursion depth is $O(\log n/\log\log n)$, and so
$Q_{2,d}(n)=O((\log n/\log\log n)\cdot \log\log n)=O(\log n)$.
Due to the constant-factor blowup, we have $S_{2,d}(n)=O(n^d 2^{O(\log n/\log\log n)})$.
However, the space bound can be reduced by applying the hierarchical cutting lemma, using the sequence
$r_i=n/n_i$, $n_0=n$, $n_{i+1}=n_i/\log^\eps n$, 
and $\ell=O(\log n/\log\log n)$.
Then $S_{2,d}(n) = \OO\left(\sum_{i=0}^{\ell-1} r_{i+1}^d n_i^{d}\right)
= \OO\left(\sum_{i=0}^{\ell-1} n^d\log^{\eps d}n\right)$, which is
$\OO(n^d)$.
\end{proof}

Note the more conventional order of levels in the above multi-level data structure: the innermost level tackles the vertical sides, and the outer two levels tackle the two nonvertical sides.
Note also how fortuitously the product of the $O(\log\log n)$ tree height in Lemma~\ref{lem:semigroup2} and the
$O(\log n/\log\log n)$
tree height in Lemma~\ref{lem:semigroup3} happens to be $O(\log n)$.

\subsubsection{Subarrangement point location}

The preceding theorem does not bound the actual query time.
When the cost of point location is included, the query complexity goes up by a logarithmic factor naively.
Next we will show how to perform multiple point location operations more efficiently in \emph{constant} time per operation, reminiscent of \emph{fractional cascading}~\cite{ChazelleG86,chan2022}.  To this end, we introduce the following subproblem:


\begin{problem}[Subarrangement Point Location]
\label{prob:subarrangement}
Given a set $H$ of $n$ hyperplanes in $\R^d$, and a subset $H'\subseteq H$ of hyperplanes, 
build a data structure that can handle the following types of queries:
Given a point $p$ and the label of the face where $p$ lies in $\mathcal{A}(H)$,
output the label of the face that $p$ lies in $\mathcal{A}(H')$.
\end{problem}

The fact that $H'$ is a subset of $H$ allows for an extremely simple solution to this subproblem, by just following pointers!

\begin{observation} \label{obs:subarrangement}
  The subarrangement point location problem in $\R^d$ can be solved using a data structure with $O(n^d)$ space and preprocessing time that has $O(1)$ query time.
\end{observation}

\begin{proof}
Build the full arrangement $\mathcal{A}(H)$  and $\mathcal{A}(H')$. Since $H'\subseteq H$,
any cell $\Delta \in H$  is in a unique cell $\Delta' \in H'$,
and we can store a pointer from $\Delta$ to $\Delta'$.
As there are $O(n^d)$ cells in $\mathcal{A}(H)$, it suffices to store $O(n^d)$ pointers.
Given a query point $p$ and the label of the face $\Delta$ in $\mathcal{A}(H)$, we can follow the pointer
to find the face of $\Delta' \in \mathcal{A}(H')$ that $p$ lies in.
\end{proof}

\subsubsection{Bounding the actual query time}

We now modify Lemmas~\ref{lem:semigroup1}--\ref{lem:semigroup3} so as to bound the actual query time.  At first,  sublogarithmic query time for Lemmas~\ref{lem:semigroup1}--\ref{lem:semigroup2} seems impossible, but we show that it is possible 
if we 
are given  the faces in the dual arrangement containing the dual query points.

\begin{lemma} \label{lem:semigroup1:new}
  There is a data structure for simplex range searching in $\R^d$ in the semigroup setting with $O(n^{d+\eps})$ space and preprocessing time and $O(1)$ query time, assuming
  that we are given the labels of the faces in $\AAA(H)$ containing $h_1^*,\ldots,h_{d+1}^*$, where $H$ denotes the dual hyperplanes of the input points, and $h_1^*,\ldots,h_{d+1}^*$ denote the dual points of the hyperplanes bounding the query simplex.
\end{lemma}
\begin{proof}
We modify the proof of Lemma~\ref{lem:semigroup1}, which is based on the method in Section~\ref{sec:prev}.
Observe that from the $d+1$ labels in $\AAA(H)$,
we can determine the labels in
$\AAA(H_\Delta)$, $\AAA(C_\Delta^+)$, and
$\AAA(C_\Delta^-)$ in constant time by Observation~\ref{obs:subarrangement}.
Thus, when we recurse in $H_\Delta$,  $C_\Delta^+$, or $C_\Delta^-$, the assumption remains true.
In addition, for each face $f$ in $\AAA(H)$,
we store a pointer from its label to 
a cell $\Delta$ in the cutting that overlaps with $f$.  If there is more than one cell in the cutting, pick one arbitrarily.

Recall that in the algorithm to answer a level-$j$ query, we need to find the cell $\Delta$ containing $h^*$, the dual point of the hyperplane bounding one of the query's halfspaces.  By assumption, we are given the label of the face $f$ in $\AAA(H)$ containing $h^*$.  We just follow the pointer from $f$ to identify $\Delta$ in $O(1)$ time.

However, there is one subtlety: the cell $\Delta$ contains a point in $f$, but does not
necessarily contain the point $h^*$.  But it doesn't matter!  All points in $f$ are equivalent, in the sense that if $h^*$ were to change to any point in the same face $f$, the answer to the query would be the same.  (The algorithm does not need the actual coordinates of $h^*$ anyway, just the label of $f$.)

The query time thus satisfies the same recurrence as before.
Since Observation~\ref{obs:subarrangement} and the extra pointers require
$O(n^d)$ space,
the space recurrence has an extra $O(n^d)$ term:
\[ S_j(n) \:=\: O(r^d)S_j(n/r) + O(r^d) S_{j-1}(n) + O(nr^d) + O(n^d). \]
As we have chosen $r=N^\eps$, we still get $S_{d+1}(N)=O(N^{d+O(\eps)})$ (and the query bound is the same).
\end{proof}

\begin{lemma}\label{lem:semigroup2:new}
  There is a data structure for $(1,d)$-sided queries in $\R^d$ in the semigroup setting with $\OO(n^{d})$ space and preprocessing time and $O(\log\log n)$ query time, assuming
  that we are given the labels of the faces in $\AAA(H)$ containing $h^*,h_1^*,\ldots,h_d^*$, where $H$ denotes the dual hyperplanes of the input points, $h^*$ denotes the dual point of the hyperplane bounding the query's nonvertical halfspace, and $h_1^*,\ldots,h_d^*$ denote the dual points of the hyperplanes bounding the vertical projection of the query's vertical halfspaces.
\end{lemma}
\begin{proof}
We modify the proof of Lemma~\ref{lem:semigroup2}, which is based on the proof of Theorem~\ref{thm:simplex_group_polylog}, by the same idea as in the proof of Lemma~\ref{lem:semigroup1:new}.
(Note that when we take the vertical projection of the input point set, we are taking a $(d-1)$-dimensional slice of the dual arrangement; we can add a pointer from the label of each face of the arrangement to the label of a corresponding face of the $(d-1)$-dimensional arrangement, if it exists.)
The space recurrence again has an extra $O(n^d)$ term:
  \[ S_{1,d}(n) \:=\: O(r^d)S_{1,d}(n/r) + O(r^d) S_{0,d}(n) + O(nr^d) + O(n^d), \]
with   $S_{0,d}(n)=O(n^{d-1+\eps})$ and $Q_{0,d}(n)=O(1)$ by Lemma~\ref{lem:semigroup1:new}.  
As we have chosen $r=n^\eps$,
we still get $S_{1,d}(n) = O(n^d 2^{O(\log\log n)})=\OO(n^d)$ (and the query bound is the same).
\end{proof}

\begin{theorem}\label{thm:simplex_semigroup}
  Simplex range searching on $n$ points in $\R^d$ with weights from a semigroup can be performed with 
  $O(\log n)$ query time using a data structure with $\OO(n^d)$ space and preprocessing time.
\end{theorem}
\begin{proof}
We modify the proof of Lemma~\ref{lem:semigroup3}, which is based on the proof of Theorem~\ref{thm:simplex_group_polylog}, by the same idea as in the proof of Lemma~\ref{lem:semigroup1:new}.
The space recurrence again has an extra $O(n^d)$ term:
  \[ S_{2,d}(n) \:=\: O(r^d)S_{2,d}(n/r) + O(r^d) S_{1,d}(n) + O(nr^d) + O(n^d), \]
with   $S_{1,d}(n)=O(n^{d-1+\eps})$ and $Q_{1,d}(n)=O(\log\log n)$ by Lemma~\ref{lem:semigroup2:new}.  
The space bound remains  $\OO(n^d)$ when using hierarchical cuttings (and the query bound is the same).

At the beginning, we can satisfy the assumption in Lemma~\ref{lem:semigroup2:new} by performing $d+1$ initial
point location queries in the global hyperplane arrangement.  This requires an additional $O(\log n)$ query time and $O(n^d)$ space~\cite{Chazelle93}.
\end{proof}

We could use hierarchical cuttings to remove the extra logarithmic factors in 
the space bound of Lemma~\ref{lem:semigroup2:new},
but there is an extra $\log^{1+O(\eps)}n$ factor in the proof of Theorem~\ref{thm:simplex_semigroup}.
Note that unlike in the remark in Section~\ref{sec:hier}, we do not see any way to improve the space bound of the above data structure to $O(n^d/\log^{\Omega(1)}n)$, because it explicitly works with arrangements of $O(n^d)$ size.


\section{Simplex Range Stabbing}\label{sec:stab}
Given a set of $n$ simplices in $\R^d$, the goal of simplex range stabbing is to build data structures so that we can quickly count the simplices that are stabbed by (i.e., contain) a query point, or report them, or compute the sum of their weights from a group or semigroup.

\subsection{Preliminaries}\label{sec:stab:prelim}
We begin by reviewing known techniques used in previous data structures.
Given a set $P$ of $n$ points in $\R^d$, a \emph{simplicial partition} 
for $P$ is a collection $\Pi = \{(P_1, \Delta_1), ..., (P_t, \Delta_t)\}$ where $P$ is a disjoint union of the $P_i$'s and each $\Delta_i$ is a simplex containing $P_i$. We will refer to the $\Delta_i$'s as the \emph{cells}.
\Matousek{}~\cite{Matousek92} proved the following key theorem:
\begin{theorem}[Partition Theorem]
  For a set $P$ of $n$ points in $\R^d$, for any parameter $t \le n$, there exists a simplicial partition $\Pi = \{(P_1, \Delta_1), ..., (P_{O(t)}, \Delta_{O(t)})\}$ such that each subset contains at most $n/t$ points, and any hyperplane crosses $O(t^{1-1/d})$ cells.  The partition can be constructed in $O(n\log t)$ time if $t\le n^\alpha$ for some constant $\alpha>0$.
\end{theorem}


The partition theorem is useful for divide-and-conquer for getting linear-space data structures with around $O(n^{1-1/d})$ query time. 
Analogously to Lemma~\ref{lem:extra_large_space_ds}, we can use the partition theorem recursively to create a multi-level data structure for simplex stabbing as follows:

Consider the case when all the input ranges are defined by $j$ halfspaces.
Call this a ``level-$j$'' problem. Like before, a level-$0$ problem can be solved trivially by storing the sum of the weights of all input ranges.
Detecting if a query point $p$ stabs a range $\gamma$ corresponds to testing if $p$ lies within each of the $j$ halfspaces defining $\gamma$. 
Dualizing, this corresponds to checking if the dual hyperplane $p^*$ lies on the correct side of the dual points of each halfspace defining $\gamma$.
Apply the partition theorem to the set $P$ of dual points of the $j$-th halfspaces of the ranges (we may assume that all of these are lower halfspaces, since afterwards we can repeat this process for all the upper halfspaces as well). For every cell $\D_i$ of the partition intersecting the hyperplane $p^*$, we can recurse on the ranges corresponding to the subset $P_i$.  For every cell $\D_i$ completely below $p^*$, we recurse  on the subset of ranges corresponding to subset $P_i$
but as a level-$(j-1)$ problem. This gives the following space and query time bounds:
\[ S_j(n) \:=\: \max_{n_i\le n/t,\ \sum_i n_i = n} \sum_i S_j(n_i) + O(t)S_{j-1}(n/t) + O(t) \]
\[ Q_j(n) \:=\: O(t^{1-1/d}) Q_j(n/t) + O(t)Q_{j-1}(n/t) + O(t). \]
(For the base case, $S_0(n) = O(1)$ and $Q_0(n) = O(1)$.)
Choosing $t = N^\eps$ where $N$ is the global input size  ensures that the recursion depth is $O(1)$.
Thus, $S_{d+1}(n)=O(N)$ and $Q_{d+1}(N)=O(N^{1-1/d+O(\eps)})$.
(The preprocessing time analysis is similar.)
This gives us the following lemma, which we  will use later in $d-1$ dimensions.

\begin{lemma} \label{lem:simplex_range_stabbing_slow}
  Simplex range stabbing on $n$ simplices in $\R^d$ with weights from a semigroup  can be performed with $O(n^{1-1/d+\eps})$ query time using a data structure with $O(n)$ space and $O(n\log n)$ preprocessing time. 
\end{lemma}

\subsection{Group simplex range stabbing}
\label{sec:stab:group}

We now present our new data structure for simplex range stabbing in the group setting
(which in particular is sufficient for counting).  We follow an approach similar to our data structure for group simplex range searching in Section~\ref{sec:group}.

Recall Definition~\ref{def:simplex_level2} on $(i,j)$-sided ranges.  An \emph{$(i,j)$-sided stabbing problem} will now refer to the case when all input ranges are $(i,j)$-sided.
As in Observation~\ref{obs1}, the simplex stabbing problem reduces to a constant number of $(2,d)$-sided stabbing problems.  As in Observation~\ref{obs2}, in the group setting,
the $(2,d)$-sided stabbing problem reduces to the
$(1,d)$-sided stabbing problem.

\begin{theorem} \label{thm:simplex_range_stabbing}
  Simplex range stabbing on $n$ simplices in $\R^d$ with  weights from a group can be performed with $\OO(n^{1-1/d})$ query time using a data structure with $O(n\log\log n)$ space and $O(n\log n)$ preprocessing time.
\end{theorem}
\begin{proof}
By the two observations, it suffices to solve the $(1,d)$-sided stabbing problem.

Like before,
by applying the partition theorem to the dual points of the nonvertical input halfspaces,
we obtain the following recurrences for the space and query time for the $(1,d)$-sided stabbing problem:
\[ S_{1,d}(n) \:=\: \max_{n_i\le n/t,\ \sum_i n_i = n} \sum_i S_{1,d}(n_i) + O(t)S_{0,d}(n/t) + O(t) \]
\[ Q_{1,d}(n) \:=\: O(t^{1-1/d}) Q_{1,d}(n/t) + O(t)Q_{0,d}(n/t) + O(t), \]
with $S_{0,d}(n)=O(n)$ and $Q_{0,d}(n)=O(n^{1-1/(d-1)+\eps})$ by Lemma~\ref{lem:simplex_range_stabbing_slow}.

We choose $t=n^\eps$ for a sufficiently small constant $\eps$.
The recursion depth is $O(\log\log n)$.
Thus, $S_{1,d}(n)=O(n\log\log n)$.  
(The preprocessing time analysis is similar, except that we get a geometric series.) 
Note that the $O(t)Q_{0,d}(n/t)$ term is $n^{1-1/(d-1)+O(\eps)}\ll n^{1-1/d}$.
Due to a constant-factor blowup,
we get $Q_{1,d}(n) = O(n^{1-1/d}2^{O(\log\log n)}) = \OO(n^{1-1/d})$.
\end{proof}

\subsection{Reducing query time from $\OO(n^{1-1/d})$ to $O(n^{1-1/d})$}\label{sec:stab:improv}

In the proof of Theorem~\ref{thm:simplex_range_stabbing}, we get an extra polylogarithmic factor in the query time because of the constant-factor blowup, but this can be avoided by using a ``hierarchical'' variant of the partition theorem, which follows from Chan's optimal partition trees \cite{Chan12} (see also \cite{Matousek93}). 
\begin{theorem}[Hierarchical Partition Theorem]\label{thm:part}
  Let $P$ be a set of $n$ points in $\R^d$. For any sequence
  $t_1 < t_2<\cdots < t_\ell\le n$ with $t_1\ge\log^{\omega(1)}n$,
  we can compute a tree of cells, such that the cell of each node at depth $i$ is the disjoint union of the cells of its $O(t_{i+1}/t_i)$ children, and
  the cells at each depth $i$ form a simplicial partition of $P$ into $O(t_i)$ cells such that each cell contains at most $n/t_i$ points and any hyperplane crosses $O(t_i^{1-1/d})$ cells.
  The tree can be constructed by a randomized algorithm in $O(n\log n)$ time, where the crossing number bound holds w.h.p.\footnote{With high probability, i.e., probability $1-O(1/n^c)$ for an arbitrarily large constant $c$.}
\end{theorem}

The original version of Chan's optimal partition tree yields the above theorem for a suffix of the sequence $1,2,4,8,\ldots$\ \
(See \cite[proof of Theorem~5.3]{Chan12}, where each node has constant degree, and a cell at depth $i$ that contains fewer than $n/2^{i+1}$ points need not be subdivided and may be viewed as a degree-1 node.)  The above generalization follows by rounding each $t_i$ to a power of 2, and keeping only the nodes with depths in the resulting subsequence.

To reduce the query time in our data structure,
we just replace all the simplicial partitions in the proof of Theorem~\ref{thm:simplex_range_stabbing} with the tree of partitions from the hierarchical partition theorem, using the sequence $t_i=n/n_i$, $n_0=n$, $n_{i+1}=n_i^{1-\eps}$, and $\ell=O(\log\log n)$.  The query time bound becomes
$Q_{1,d}(n)=O\left(\sum_{i=0}^{\ell-1} (t_{i+1}/t_i)\cdot t_i^{1-1/d} n_i^{1-1/(d-1)+\eps}\right)
= O\left(\sum_{i=0}^{\ell-1}
n^{1-1/d}/ n_i^{1/(d-1)-1/d-O(\eps)}\right) =
O(n^{1-1/d})$.

\subsection{Reducing space from $O(n\log\log n)$ to $O(n)$}\label{sec:stab:bit}

We next show that the extra $\log\log n$ factor in the space bound of Theorem~\ref{thm:simplex_range_stabbing} can also be removed.

We assume a real-RAM model of computation where a word
can store (i)~a real number, (ii)~a group element, or
(iii)~a $w$-bit number with $w\ge\log n$.
Words of type (iii) are standard, since indices and pointers require
logarithmically many bits.  Words of type~(i) (which are commonly assumed in computational geometry) actually pose more of a concern, as we could potentially cheat by hiding an unlimited number of bits inside a real number.  Our algorithms will not cheat.  One way to prevent cheating in the model is to insist that each real number stored must be obtained by one of the standard arithmetic operations on two other real numbers that are stored in the data structure or are among the input real numbers (in particular, the floor function is not allowed).

We will use bit packing tricks but only on words of type (iii), which are legal (although our data structure does not fit in the standard pointer machine model, it fits in what Chazelle called an ``arithmetic pointer machine''~\cite{Chazelle86filter} where arithmetic operations on addresses are allowed).
These tricks are commonly used in orthogonal range searching and related geometric problems (e.g., see \cite{ChanNRT18}), and also in other areas of data structures.
We will \emph{not} use any bit packing on words of type~(ii), i.e., each group element are treated as ``atomic''; our linear-space data structures will store really only $O(n)$ number of group elements.

Our key observation is that by bit packing,
the data structure in
Lemma~\ref{lem:simplex_range_stabbing_slow}
actually takes \emph{sublinear} words of extra space if $n$ is small, and if space for the input array is excluded.

\begin{lemma} \label{lem:simplex_range_stabbing_slow:bit}
  Simplex range stabbing on $n$ simplices in $\R^d$ with weights from a semigroup can be performed with $O(n^{1-1/d+\eps})$ query time using a data structure with $O((n\log n)/w + n^{1-\Omega(\eps)})$ words of space and $O(n\log n)$ preprocessing time.  The space bound here excludes the input array storing the $n$ simplices and their weights (we are not allowed to permute the input array).
\end{lemma}
\begin{proof}
We reanalyze the space bound in the recursive method in Section~\ref{sec:stab:prelim}.
Recall that we have chosen $t=N^\eps$ where $N$ is the global input size.
When $n=t$, we have reached the base case and can simply store the input as a plain list of pointers to the input array, which requires $O(n\log N)$ bits, i.e., $O((n\log N)/w)$ words.
As the recursion depth is $O(1)$ and there are $O(N/t)$ leaves (and $O(N/t^2)$ internal nodes), we have
$S_{d+1}(N) = O((N\log N)/w + N/t) = O((N\log N)/w + N^{1-\eps})$.
\IGNORE{

We modify the recursive method in Section~\ref{sec:stab:prelim}, except that when the input size is below some parameter $B$, we simply store the input as a plain list of pointers to the input array, which requires $O(n\log n)$ bits, i.e., $O((n\log n)/w)$ words, and answer a query by naive linear search in $O(B)$ time.

Thus, if $n\ge B$, we still have
\[ S_j(n) \:=\: \max_{n_i\le n/t,\ \sum_i n_i = n} \sum_i S_j(n_i) + O(t)S_{j-1}(n/t) + O(t) \]
\[ Q_j(n) \:=\: O(t^{1-1/d}) Q_j(n/t) + O(t)Q_{j-1}(n/t) + O(t), \]
with $S_0(n)=O(1)$ and $Q_0(n)=O(1)$.
But if $n\le B$, we have $S_j(n)=O((n\log n)/w)$ and
$Q_j(n)=O(B)$.

As we have chosen $t=N^\eps$ to ensure that the recursion depth is $O(1)$,
the recurrences solve to
$S_{d+1}(N) = O((N\log N)/w + N/B)$
and $Q_{d+1}(N) = O(BN^{1-1/d+\eps})$.
We set $B=N^\eps$ (and readjust~$\eps$).
}
\end{proof}

\begin{theorem} \label{thm:simplex_range_stabbing:new}
  Simplex range stabbing on $n$ simplices in $\R^d$ with weights from a group can be performed in $O(n^{1-1/d})$ query time w.h.p.\  using a randomized data structure with $O(n)$ words of space and $O(n\log n)$ preprocessing time.
\end{theorem}
\begin{proof}
We modify the proof of Theorem~\ref{thm:simplex_range_stabbing}.
We choose a favorable permutation of the input array: namely, having computed the simplicial partition $\{(P_1,\Delta_i),\ldots,
(P_{O(t)},\Delta_{O(t)})\}$, we permute the array so that each $P_i$ occupy a contiguous subarray, before recursing in each $P_i$; afterwards, the permutation of $P_i$ is fixed, and we can apply Lemma~\ref{lem:simplex_range_stabbing_slow:bit}.

By the lemma, we can plug in $S_{0,d}(n) = O((n\log n)/w + n^{1-\Omega(\eps)})$, and 
the recurrence for the amount of extra space (excluding the input array) becomes
\[ S_{1,d}(n) \:=\: \max_{n_i\le n/t,\ \sum_i n_i = n} \sum_i S_{1,d}(n_i) + O((n\log n)/w + t(n/t)^{1-\eps}).  \]
As we have chosen $t=n^\eps$, 
this gives $S_{1,d}(n) = O((n\log n)/w + n)$
(since the $(n\log n)/w$ term generates a geometric series, and
the $t(n/t)^{1-\eps}$ term is sublinear).
Since $w\ge\log n$, we obtain a linear space bound.
The query time is $O(n^{1-1/d})$ as already  explained in Section~\ref{sec:stab:improv} by using the hierarchical partition theorem. 
\end{proof}

Note that in the above theorem (and, in fact, all data structures in this paper), the only place where randomization is used is the construction of the hierarchical partitions (from Theorem~\ref{thm:part}).  If we do not care about preprocessing time, all our data structures are deterministic.

\paragraph{Remark.}
For counting, it is possible to use bit packing tricks to slightly improve the query time as well.
In modifying the proof of Theorem~\ref{thm:simplex_range_stabbing},
we apply the hierarchical partition theorem with
$t_\ell=n/B$ for some parameter $B$.  We can switch to a different data structure for leaf subproblems of size $B$, 
with $O(m)$ space and
$O(B^{1+\eps}/m^{1/d})$ query time~\cite{Matousek93}.
By choosing $m=B^{1+\delta}$ for a sufficiently small constant $\delta>0$,
the query time is $O(B^{1-1/d-\delta/d+\eps})$
and the space usage is $O(B^{1+\delta}\log B)$ in bits, or $O((B^{1+\delta}\log B)/w)$ in words by bit packing.
The total space in words is $O(n+(n/B)\cdot (B^{1+\delta}\log B)/w)$, and the query time is $O((n^{1-1/d}/B^{1/(d-1)-1/d-O(\eps)} + (n/B)^{1-1/d}\cdot B^{1-1/d-\delta/d+\eps})$.
Setting $B$ near $\log^{1/\delta}n$
gives $O(n)$ words of space and
$O(n^{1-1/d}/\log^{1/d-O(\eps)} n)$ query time (as $w\ge\log n$).

We are not aware any previous work mentioning this type of
$O(n^{1-1/d}/\log^{\Omega(1)}n)$ query time bound.
The same trick works also for simplex range counting or reporting.  This trick does not work in the group setting
(since group elements requires one word each and cannot be packed).

\subsection{Simplex range stabbing reporting}\label{sec:stab:report}

For the reporting version of simplex range stabbing, we follow an approach similar to our data structure for simplex range reporting in Section~\ref{sec:report}.

We start with the $(1,0)$-sided stabbing problem, which reduces to halfspace range reporting in dual space.  By known results~\cite{MatousekS93}, we have:

\begin{lemma}\label{lem:halfspace_range_reporting2}
Halfspace reporting on $n$ points in $\R^d$ 
can be performed with $O(n^{1-1/\lfloor d/2\rfloor +\eps}+k)$ query time using a data structure
with $O(n)$ space and $O(n\log n)$ preprocessing time for any fixed $\eps>0$.
\end{lemma}

We then solve the $(1,d)$-sided stabbing problem and finally the $(2,d)$-sided stabbing problem:

\begin{lemma} \label{lem:stab:report1}
There is a data structure for $(1,d)$-sided stabbing problem in $\R^d$ with $O(n^{1-1/(d-1)+\eps}+k)$ query time, $O(n)$ space, and $O(n\log n)$ preprocessing time for any fixed $\eps > 0$.
\end{lemma}
\begin{proof}
Applying the same method as in Section~\ref{sec:stab:prelim} using simplicial partitions in $d-1$ dimensions, we obtain a data structure for the $(1,j)$-sided stabbing problem
with the following recurrences of space and query time (ignoring the ``$+k$'' reporting cost):
\[ S_{1,j}(n) \:=\: \max_{n_i\le n/t,\ \sum_i n_i = n} \sum_i S_{1,j}(n_i) + O(t)S_{1,j-1}(n/t) + O(t) \]
\[ Q_{1,j}(n) \:=\: O(t^{1-1/(d-1)}) Q_{1,j}(n/t) + O(t)Q_{1,j-1}(n/t) + O(t), \]
with $S_{1,0}(n)=O(n)$ and $Q_{1,0}(n)=O(n^{1-1/\lfloor d/2\rfloor +\eps})$ by Lemma~\ref{lem:halfspace_range_reporting2}.

We choose $r=N^\eps$ where $N$ is the global input size, to ensure that the recursion depth is $O(1)$.
This gives $S_{1,d}(N)=O(N)$ and
$Q_{1,d}(N)=O(N^{1-1/(d-1)+O(\eps)})$.
(The preprocessing time analysis is similar.)
\end{proof}

\begin{theorem} \label{thm:simplex_stab_reporting}
Simplex stabbing reporting on $n$ simplices in $\R^d$
can be performed with  $O(n^{1-1/d}+k)$ query time w.h.p.\ (where $k$ is the output size) using a randomized data structure with $O(n)$ words of space and $O(n\log n)$ preprocessing time.
\end{theorem}
\begin{proof}
It suffices to solve the $(2,d)$-sided stabbing problem:
As in the proof of Theorem~\ref{thm:simplex_range_stabbing}, we obtain a data structure with the following recurrences for space and query time (ignoring the ``$+k$'' reporting cost):
\[ S_{2,d}(n) \:=\: \max_{n_i\le n/t,\ \sum_i n_i = n} \sum_i S_{2,d}(n_i) + O(t)S_{1,d}(n/t) + O(t) \]
\[ Q_{2,d}(n) \:=\: O(t^{1-1/d}) Q_{2,d}(n/t) + O(t)Q_{1,d}(n/t) + O(t), \]
with $S_{1,d}(n)=O(n)$ and $Q_{1,d}(n)=O(n^{1-1/(d-1)+\eps})$ by Lemma~\ref{lem:simplex_range_stabbing_slow}.
(The preprocessing time analysis is similar.)

We choose $r=n^\eps$.  As in the proof of
Theorem~\ref{thm:simplex_range_stabbing}, 
the recurrences solve to $S_{2,d}(n)=O(n\log\log n)$ and
$Q_{2,d}(n)=\OO(n^{1-1/d})$.

As in Section~\ref{sec:stab:improv},
the extra polylogarithmic factor in the query time bound can be removed, by using the hierarchical partition theorem for the outermost level.

As in Section~\ref{sec:stab:bit},
the extra $\log\log n$ factor in space can also be improved, by the same bit-packing tricks (by observing that the space bounds
in Lemmas~\ref{lem:halfspace_range_reporting2} and~\ref{lem:stab:report1} can be reduced to
$O((n\log n)/w + n^{1-\eps'})$ for some $\eps'>0$).
\end{proof}

\section{Segment Intersection Searching and Ray Shooting}\label{sec:seg}

Given a set $S$ of $n$ line segments in $\R^2$,
the goal of segment intersection searching is to quickly count the segments in $S$ intersecting a query line segment, or report them, or compute the sum of their weights from a group or semigroup.  Naively, the condition that an input segment $s$ intersects a query segment $q$ can be expressed as a conjunction of four 2D halfplane (linear) constraints: namely, that the two endpoints of $q$ lie on different sides of the line through $s$ \emph{and} the two endpoints of $s$ lie on different sides of the line through $q$.  Thus, it is straightforward to obtain a multi-level partition tree for this problem achieving near linear space and near $\sqrt{n}$ query time; however, the four levels cause multiple extra logarithmic factors in space and time.  We will describe better approaches, using expressions involving fewer 2D halfplane constraints.

\subsection{Group segment intersection searching}

In this subsection, we consider segment intersection searching queries in the group setting (which in particular is sufficient for counting).

\begin{observation}
In the group setting,
a line-segment intersection query for line segments reduces to
$O(1)$ rightward-ray intersection queries for  rightward rays.
\end{observation}
\begin{proof}
By subtraction, a line-segment intersection query for $n$ line segments reduces to a line-segment intersection query for two sets of $n$ rightward rays.
By subtraction again, a line-segment intersection query reduces to two rightward-ray intersection queries on the same input set.
\end{proof}

By the above observation, it suffices to solve the problem in the case where all input segments and all query segments are rightward rays.

For a rightward ray $s$, let $\ell(s)$ denote the line through $s$,
let $p(s)$ denote the initial point of $s$, 
let $x(s)$ denote the $x$-coordinate of $p(s)$, and
let $m(s)$ denote the slope of $s$.

There are 4 possible ways in which two rightward rays $s$ and $q$ may intersect (ignoring degeneracies):
\begin{itemize}
\item {\sc Type A:} $x(s)>x(q)$, $m(s)>m(q)$, and $p(s)$ is below $\ell(q)$.
\item {\sc Type A$'$:} $x(s)>x(q)$, $m(s)<m(q)$, and $p(s)$ is above $\ell(q)$.
\item {\sc Type B:} $x(s)<x(q)$, $m(s)>m(q)$, and $p(q)$ is above $\ell(s)$.
\item {\sc Type B$'$:} $x(s)<x(q)$, $m(s)<m(q)$, and $p(q)$ is below $\ell(s)$.
\end{itemize}
\begin{figure}
    \centering
    \includegraphics[page=2,width=0.24\textwidth]{rays_swapped.pdf}
    \includegraphics[page=1,width=0.24\textwidth]{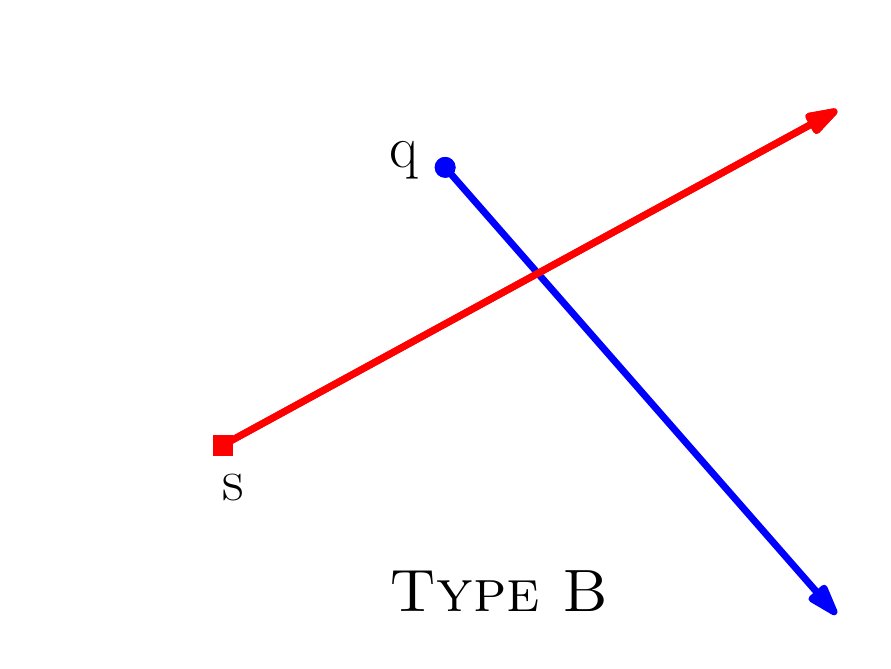}
    \caption{}
    \label{fig:rays}
\end{figure}

See Figure~\ref{fig:rays}.  We focus on counting type A intersections, since all the other types of intersections can be handled in a similar manner
(for example, in type B, $p(q)$ is above $\ell(s)$ iff $\ell(s)^*$ is above $p(q)^*$ by duality).
We follow an approach similar to our data structure for group simplex range stabbing  in Section~\ref{sec:stab:group}.

\begin{definition}\label{def:stab:group}
Let $S$ be the input set of rightward rays.
Define the following types of queries for a given query rightward ray $q$:
\begin{itemize}
\item \emph{level-1 query}: compute the sum of the weights of all $s\in S$ such that $x(s)>x(q)$ and $m(s)>m(q)$.
\item \emph{level-2 query}: compute the sum of the weights of all $s\in S$ such that $x(s)>x(q)$, $m(s)>m(q)$, and $p(s)$ is below $\ell(q)$.
\end{itemize}
\end{definition}

\begin{lemma} \label{lem:seg_group1}
There is a data structure for  level-1 queries as defined above with $O(n^\eps)$ query time, $O(n)$ space, and $O(n\log n)$ preprocessing time for any fixed $\eps > 0$.
\end{lemma}
\begin{proof}
Level-1 queries reduce to 2D  orthogonal range searching (by treating $(x(s),m(s))\in \R^2$ as an input point and $(x(q),\infty)\times (m(q),\infty)$ as a query range).  We can use a standard range tree~\cite{BergCKO08} with branching factor $N^\eps$, where $N$ denotes the size of the global point set.
\end{proof}

\begin{theorem} \label{thm:seg_group}
  Intersection searching for $n$ line segments in $\R^2$ with weights from a group can be performed with $O(\sqrt{n})$ query time w.h.p.\ using a randomized data structure with $O(n)$ space 
  and $O(n\log n)$ preprocessing time.
\end{theorem}
\begin{proof}
By applying the partition theorem to the point
set $\{p(s): s\in S\}$, we obtain the following recurrences for the space and query time for level-2 queries:
\[ S_2(n) \:=\: \max_{n_i\le n/t,\ \sum_i n_i = n} \sum_i S_2(n_i) + tS_1(n/t) + O(n) \]
\[ Q_2(n) \:=\: O(\sqrt{t})Q_2(n/t) + O(t)Q_1(n/t) + O(t), \]
with $S_1(n)=O(n)$ and  $Q_1(n)=O(n^\eps)$ by Lemma~\ref{lem:seg_group1}.
(The preprocessing time analysis is similar.)

We choose $t=n^\eps$.
The recursion depth is $O(\log\log n)$. 
Note that the $O(t)Q_1(n/t)$ term is $n^{O(\eps)}\ll \sqrt{n}$.
This gives $S_2(n)=O(n\log\log n)$ and $Q_2(n)=\OO(\sqrt{n})$.

As in Section~\ref{sec:stab:improv},
the extra polylogarithmic factor in the query time bound can be removed, by using the hierarchical partition theorem for the outer level.

As in Section~\ref{sec:stab:bit},
the extra $\log\log n$ factor in space can also be improved, by the same bit-packing tricks.
\end{proof}

We note that Bar-Yehuda and Fogel's work from the 1990s~\cite{Bar-YehudaF94} actually already used the same subtraction trick and similar ideas with comparisons of $x$-coordinates and slopes, but their multi-level data structure used a more conventional order of the levels that resulted in multiple extra logarithmic factors.

\subsection{Segment intersection reporting}

For segment intersection reporting, the subtraction trick is no longer applicable.  We follow an approach similar to our data structure for simplex range stabbing reporting in Section~\ref{sec:stab:report}.

For a line segment $s$, let $\ell(s)$ denote the line through $s$, 
let $p_L(s)$ (resp.\ $p_R(s)$) denote the left (resp.\ right) endpoint of $s$, and
let $x_L(s)$ (resp.\ $x_R(s)$) denote the $x$-coordinate of $p_L(s)$ (resp.\ $p_R(s)$).

There are 8 possible ways in which
two line segments $s$ and $q$ may intersect (ignoring degeneracies):
\begin{itemize}
\item {\sc Type A:} $x_L(s)<x_L(q)<x_R(s)<x_R(q)$,
$p_L(q)$ is above $\ell(s)$,
and 
$p_R(s)$ is above $\ell(q)$.
\item {\sc Type A$'$:} $x_L(s)<x_L(q)<x_R(s)<x_L(q)$,
$p_L(q)$ is below $\ell(s)$,
and 
$p_R(s)$ is below $\ell(q)$.
\item {\sc Type B:} $x_L(q)<x_L(s)<x_R(q)<x_L(s)$,
$p_L(s)$ is below $\ell(q)$,
and 
$p_R(q)$ is below $\ell(s)$.
\item {\sc Type B$'$:} $x_L(q)<x_L(s)<x_R(q)<x_L(s)$,
$p_L(s)$ is above $\ell(q)$,
and 
$p_R(q)$ is above $\ell(s)$.
\item {\sc Type C:} $x_L(s)<x_L(q)<x_R(q)<x_R(s)$,
$p_L(q)$ is above $\ell(s)$,
and 
$p_R(q)$ is below $\ell(s)$.
\item {\sc Type C$'$:} $x_L(s)<x_L(q)<x_R(q)<x_R(s)$,
$p_L(q)$ is below $\ell(s)$,
and
$p_R(q)$ is above $\ell(s)$.
\item {\sc Type D:} $x_L(q)<x_L(s)<x_R(s)<x_R(q)$,
$p_L(s)$ is below $\ell(q)$, 
and 
$p_R(s)$ is above $\ell(q)$.
\item {\sc Type D$'$:} $x_L(q)<x_L(s)<x_R(s)<x_R(q)$,
$p_L(s)$ is above $\ell(q)$,
and 
$p_R(s)$ is below $\ell(q)$.
\end{itemize}
\begin{figure}
    \centering
    \includegraphics[page=1,width=0.24\textwidth]{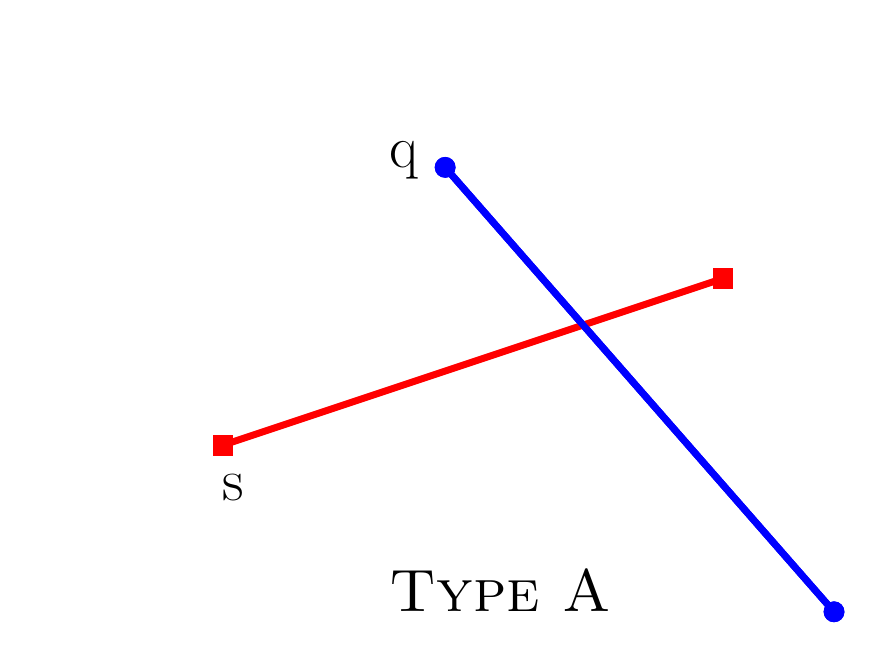}
    \includegraphics[page=2,width=0.24\textwidth]{line_seg_old.pdf}
    \includegraphics[page=3,width=0.24\textwidth]{line_seg_old.pdf}
    \includegraphics[page=4,width=0.24\textwidth]{line_seg_old.pdf}
    \caption{}
    \label{fig:line_seg}
\end{figure}

See Figure~\ref{fig:line_seg}.
We focus on reporting type A intersections, since all the other types of intersections can be handled in a similar manner (as all of these involve three 1D constraints along with two 2D halfplane constraints after the appropriate dualizations).

\begin{definition}
Let $S$ be the input set of line segments.
Define the following types of queries for 
a query segment $q$:
\begin{itemize}
\item \emph{level-1 query}: report all $s\in S$ such that $p_L(q)$ is above $\ell(s)$;
\item \emph{level-2 query}: report all $s\in S$ such that $p_L(q)$ is above $\ell(s)$ and $x_L(s)<x_L(q)$;
\item \emph{level-3 query}: report all $s\in S$ such that $p_L(q)$ is above $\ell(s)$ and $x_L(s)<x_L(q)<x_R(s)$;
\item \emph{level-4 query}: report all $s\in S$ such that $p_L(q)$ is above $\ell(s)$ and $x_L(s)<x_L(q)< x_R(s)< x_R(q)$;
\item \emph{level-5 query}: report all $s\in S$ such that $p_L(q)$ is above $\ell(s)$, $x_L(s)< x_L(q)< x_R(s)< x_R(q)$, and $p_R(s)$ is above $\ell(q)$.
\end{itemize}
\end{definition}

\begin{lemma} \label{lem:seg_report1}
There is a data structure for level-1 queries as defined above with $O(\log n+k)$ query time,
$O(n)$ space, and $O(n\log n)$ preprocessing time.
\end{lemma}
\begin{proof}
Level-1 queries are just halfplane range reporting queries in two dimensions~\cite{ChazelleGL85} in the dual.
\end{proof}

\begin{lemma} \label{lem:seg_report4}
There is a data structure for level-4 queries as defined above with $O(n^\eps +k)$ query time,
$O(n)$ space, and $O(n\log n)$ preprocessing time for any fixed $\eps > 0$.
\end{lemma}
\begin{proof}
To reduce the level-2 query problem to the level-1 query problem,
we can just use a one-dimensional partition:
form $t$ intervals each containing $n/t$ values in $\{ x_L(s): s\in S\}$, and recurse on the corresponding $t$ subsets of size $n/t$.

Similarly, we can reduce level-3 queries to level-2 queries, and level-4 queries to level-3.  (These 3 levels of the data structure essentially correspond to a range tree~\cite{BergCKO08}.)

Thus, for $j\in\{2,3,4\}$, we have following recurrences for the space and query time for level-$j$ queries (ignoring the ``$+k$'' reporting cost):
\[ S_j(n) \:=\:  tS_j(n/t) + tS_{j-1}(n/t) + O(n) \]
\[ Q_j(n) \:=\: Q_j(n/t) + O(t)Q_{j-1}(n/t) + O(t), \]
with $S_1(n)=O(n)$ and $Q_1(n)=O(\log n)$ by Lemma~\ref{lem:seg_report1}.

We choose $t=N^\eps$ where $N$ is the global input size, to ensure that the recursion depth is $O(1)$.  This gives
$S_4(N)=O(N)$ and $Q_4(N)=O(N^{O(\eps)})$.
(The preprocessing time analysis is similar.)
\end{proof}

\begin{theorem} \label{thm:seg_report}
  Segment intersection reporting for $n$ line segments in $\R^2$
  can be performed with $O(\sqrt{n}+k)$ query time w.h.p.\ (where $k$ is the output size) using a randomized data structure with $O(n)$ space and $O(n\log n)$ preprocessing time.
\end{theorem}
\begin{proof}
By applying the partition theorem to the point
set $\{p_R(s): s\in S\}$, we obtain the following recurrences for the space and query time for level-5 queries:
\[ S_5(n) \:=\:  \max_{n_i\le n/t,\ \sum_i n_i = n} \sum_i S_5(n_i) + tS_4(n/t) + O(n) \]
\[ Q_5(n) \:=\: O(\sqrt{t})Q_5(n/t) + O(t)Q_4(n/t) + O(t), \]
with $S_4(n)=O(n)$ and $Q_4(n)=O(n^\eps)$ by Lemma~\ref{lem:seg_report4}.

We choose $t=n^\eps$.
The recursion depth is $O(\log\log n)$.  Note that the $O(t)Q_4(n/t)$ term is $n^{O(\eps)}\ll \sqrt{n}$. This gives $S_5(n)=O(n\log\log n)$ and $Q_5(n)=\OO(\sqrt{n})$.
(The preprocessing time analysis is similar.)

As in Section~\ref{sec:stab:improv},
the extra polylogarithmic factor in the query time bound can be removed, by using the hierarchical partition theorem for the outermost level.

As in Section~\ref{sec:stab:bit},
the extra $\log\log n$ factor in space can also be improved, by the same bit-packing tricks.
\end{proof}

Note the unconventional order in the above multi-level data structure (similar to our data structures for simplex range reporting and range stabbing reporting): the innermost level deals with one 2D halfplane constraint, the middle levels deal with 1D constraints, and the outermost level deals with the second 2D halfplane constraint.

\subsection{Ray shooting among line segments}

We can apply our new result for segment intersection reporting to solve the ray shooting problem for
line segments, via a simple randomized black-box reduction (e.g., see~\cite{Chan96} for a similar reduction for a different problem):

\begin{corollary} \label{cor:ray_shooting}
  There exists a randomized data structure for ray shooting among $n$ line segments $S$ in $\R^2$ with
  $O(n)$ space and $O(n\log n)$ preprocessing time such that each query takes $O(\sqrt{n})$ time w.h.p.
\end{corollary}
\begin{proof}
  Take a random subset $R\subset S$ of size $n/2$, recursively build a ray shooting data structure for $R$, and build a segment intersection reporting structure for $S$.
  To answer a ray shooting query for a ray $q$,
  we first recursively find the first point $p$ hit by $q$ in $R$.  Let $\overline{q}$ be the line segment going from the initial point of $q$ to this point $p$.  Since the interior of $\overline{q}$
  does not intersect any segments of $R$, a standard
  $\eps$-net argument~\cite{Chan96} implies that $\overline{q}$ intersects only $k=O(\log n)$ segments of $S$ w.h.p.
  We can enumerate all of these segments by the segment intersection reporting structure, and return the first one hit.
  
  Using the segment intersection reporting structure from Theorem~\ref{thm:seg_report}, we get the following recurrences for the space and query time of the new data structure:
  \[ S(n) \:=\: S(n/2) + O(n) \]
  \[ Q(n) \:=\: Q(n/2) + O(\sqrt{n}), \]
  implying that $S(n)=O(n)$ and $Q(n)=O(\sqrt{n})$.  (The preprocessing time analysis is similar.)
\end{proof}

As in the remark in Section~\ref{sec:stab:bit}, the query time can be further reduced to $O(\sqrt{n}/\log^{\Omega(1)}n)$ for intersection counting and reporting and ray shooting, by bit-packing tricks.

\appendix


\section{Trade-Offs}

We can obtain space/query-time trade-off versions of many of our new results.
In this appendix, as one example,  we consider simplex range searching in the group setting.  

\begin{theorem} \label{thm:simplex_group:tradeoff}
For a parameter $n\le m \le n^d/\log^d n$,
  simplex range searching on $n$ points in $\R^{d}$ with weights from a group can be performed with $O(n/m^{1/d})$ query time using
  a data structure with $O(m(\log\log (m/n))^{O(1)})$ space.
\end{theorem}
\begin{proof}
We split into two cases depending on the size of $m$.

If $m \ge n^d/\log^{c} n$ for an arbitrarily large constant $c\ge d$, we use the data structure from the remark in Section~\ref{sec:hier}
with $O((n^d/A^\delta)\log^{O(1)}A)$ space and
$O(\log n + A^{\delta/d})$ query time.
Choosing $A=(n^d/m)^{1/\delta}$gives our desired space and query time bounds.


If $m < n^d/\log^{\omega(1)} n$, then we switch to a partition tree instead of a cutting tree for the primary structure.  Assuming a data structure for leaves of size $B$ with space $S(B)$ and query time $Q(B)$,
Chan's optimal partition tree (see \cite[Theorems 4.2 and Corollary 7.7]{Chan12}) implies a new data structure with
space $O((n/B)S(B))$ and query time $O((n/B)^{1-1/d}Q(B))$, assuming  $B < n/\log^{\omega(1)} n$.
We plug in $S(B)=O((B^d/\log^d B)(\log\log B)^{O(1)})$
and $Q(B)=O(\log B)$ by the data structure from 
the remark in Section~\ref{sec:hier}.
Choosing $B = ((m/n)\log^d (m/n))^{1/(d-1)}$ gives our desired space and query time bounds.
\end{proof}

In contrast, the previous result by Matou\v sek~\cite{Matousek93} had
$O(m)$ space and $O((n/m^{1/d})\log^{d+1}(m/n))$ query time
for $n\le m\le n^d$.
As mentioned in the remark in Section~\ref{sec:hier}, the
extra $\log\log n$ factors can be further improved by bootstrapping.

\IGNORE{

remaining questions:

semigroup for stabbing and segment intersection

tradeoffs for stabbing and segment intersection and ray shooting

eliminating extra factors in tradeoffs?

eliminating extra log in space for semigroup?

triangle intersection: reporting ok, counting?

simplex containment: hard?

deterministic preprocessing time

}

\bibliographystyle{plain}
\bibliography{ref}

\end{document}